\documentclass{amsart}

\usepackage{amssymb, amsmath}
\usepackage{mathrsfs}
\usepackage{mathtools}
\usepackage{amscd}
\usepackage{verbatim}
\usepackage{bbm}

\usepackage{wasysym}

\usepackage[utf8]{inputenc}
\usepackage{enumerate}
\usepackage{url}

\usepackage[colorinlistoftodos,prependcaption,textsize=tiny]{todonotes}

\usepackage[colorlinks,linkcolor={blue},citecolor={blue},urlcolor={red},]{hyperref}

\allowdisplaybreaks

\theoremstyle{plain}
\newtheorem{theorem}{Theorem}
\theoremstyle{remark}

\theoremstyle{plain}
\newtheorem{corollary}[theorem]{Corollary}
\newtheorem{lemma}[theorem]{Lemma}
\newtheorem{proposition}[theorem]{Proposition}

\numberwithin{theorem}{section}

\numberwithin{equation}{section}

\newcommand{\K}{\mathrm{K}}

\def\N{{\mathbb N}}

\def\Q{{\mathbb Q}}
\def\R{{\mathbb R}}

\newcommand{\E}{{\mathbb E}}
\renewcommand{\P}{{\mathbb P}}

\newcommand{\dint}{\mathrm d}

\newcommand{\mbf}[1]{\mathbf{#1}}

\begin{document}

\title[Branching processes in multicomponent coagulation]{Gelation and localization in multicomponent coagulation with multiplicative kernel through branching processes}

\author{Jochem Hoogendijk}
\address{Mathematical institute\\
    Utrecht University\\ 3508 TA Utrecht\\ The
    Netherlands}\email{j.p.c.hoogendijk@uu.nl}

\author{Ivan Kryven}
\address{Mathematical institute\\ 
    Utrecht University\\ 3508 TA Utrecht\\ The
    Netherlands}\email{i.v.kryven@uu.nl}

\author{Camillo Schenone}\email{camillo.schenone@gmail.com}

\begin{abstract}
    The multicomponent coagulation equation is a generalisation of the Smoluchowski coagulation equation in which size of a particle is described by a vector. As with the original Smoluchowski equation, the multicomponent coagulation equation features gelation when supplied with a multiplicative kernel. Additionally, a new type of behaviour called localization is observed due to the multivariate nature of the particle size distribution. Here we extend and apply the branching process representation technique, which we introduced to study differential equations in our previous work, to find a concise  probabilistic solution of the multicomponent coagulation equation supplied with monodisperse initial conditions and provide short proofs for the gelation time and localization.\\\
    
\end{abstract}
\maketitle

\section{Introduction}
Consider the following system of a countable number of ordinary differential equations 
\begin{equation}\label{eq:smoluchowski}
    \frac{\mathrm{d}w_n}{\mathrm{d}t} = \frac{1}{2} \sum\limits_{k+l = n} \K(k, l) w_k w_l - \sum\limits_{k=1}^\infty \K(n, k) w_n w_k, \qquad n \in \mathbb{N},
\end{equation}
subject to initial conditions $\{w_n(0)\}_{n \in \N}$.
Originally introduced by Smoluchowski \cite{smoluchowski1916} as a model for evolution of the size distribution $w_n(t)$ of merging particles (also called colloids or clusters), this equation became one of the standard equations in mathematical physics. 
If we interpret $n\in \N$ as the particle size and $w_n$ as a (non-normalised) size distribution, the two terms in the right hand side of  equation \eqref{eq:smoluchowski} represent respectively the rate of production of new particles of size $n=k+l$ through merging two smaller particles of sizes $k$ and $l$, and the rate of consumption  of particles of size $n$ as they themselves are being incorporated into larger particles \cite{aldous1999,fournier2009}.
The vast popularity of equation \eqref{eq:smoluchowski} for modeling, see for example, review articles \cite{leyvraz2003,aldous1999} and the references therein, is due to the flexibility one gains by specifying the kernel $\K(k,l)$ to reflect a particular physical mechanism for merging of particles. Among the most popular are the constant $\K(k, l) = 1,$  additive $\K(k, l) = k+l$ and multiplicative $\K(k, l) = kl$ kernels for $k, l \in \N$, though kernels with fractional exponents also appear \cite{aldous1999}.

In this paper, we study a more general multicomponent ({\it i.e.} multidimensional)  Smoluchowski's equation in which the particle sizes are vector-valued and the coagulation kernel is a symmetric bilinear form,
$\K(\mathbf{k}, \mathbf{l}) = \mbf{k}^T A \mbf{l}$. This kernel is a generalization of the multiplicative kernel in the monocomponent setting. 
The multicomponent coagulation equation is then written out as:
\begin{equation}\label{eq:colored_smoluchowski}
    \frac{\mathrm{d}w_{\mathbf{n}}}{\mathrm{d}t} = \frac{1}{2} \sum\limits_{\mathbf{k}+\mathbf{l} = \mathbf{n}} \K(\mathbf{k}, \mathbf{l}) w_{\mathbf{k}} w_{\mathbf{l}} - \sum\limits_{\mbf{k} \in \mathbb{N}_0^m} \K(\mathbf{n}, \mathbf{k}) w_{\mathbf{n}} w_{\mathbf{k}}, \qquad \mathbf{n} \in \mathbb{N}_0^m,
\end{equation}
with $\{w_{\mbf{n}}(0)\}$ being the initial distribution.  
This equation has appeared in many studies with different kernels, such as additive \cite{fernandez2007}, rank-one multiplicative \cite{fernandez2010}, and strictly sublinear \cite{ferreira2021multicomponent, ferreira2021, ferreira_2022, throm2023} kernels. The multiplicative multiconponent kernel first appeared in the mathematics literature in combinatorial studies \cite{andreis_2023} and \cite{kovchegov_2023}.

Although systems \eqref{eq:smoluchowski} and \eqref{eq:colored_smoluchowski} are comprised of ordinary differential equations, the fact that they contain a countably infinite number of such equations means that their behaviour (and suitable analysis techniques) are closer to those appearing in partial differential equations (PDEs).
Moreover, as discussed in the next section, one can map these systems by using a functional transform to vector-valued inviscid Burgers' PDEs in $1+d$ dimensions, which provides a fruitful avenue for analysis of the former. Although Burgers' equation  does not have known analytical solutions even in 1+1 dimension, we have shown in our previous work \cite{hoogendijk2023} that these solutions can generally be represented using a branching process from probability theory. This approach provides a probabilistic route to analysis of the multicomponent Smoluchowski equation with multiplicative kernel, as well as yielding short proofs. The goal of the current paper is therefore to expand this representation technique to multidimensional PDEs and use it for analysis of the multicomponent Smoluchowski coagulation equation with the multiplicative kernel.


The structure of this paper is as follows. After establishing the equivalence between equation \eqref{eq:colored_smoluchowski} and multidimensional inviscid Burger's equation in Theorem \ref{thm:gf_transform}, we show that the solution of the former can be  written as an expectation with respect to a certain multitype branching process in Theorem \ref{thm:multitype_branching_process}, which constitutes the core of our analysis technique. This equivalence holds until certain critical time $T_c$, which marks the breach of the mass conservation in the coagulation equation -- a phenomenon known as the \emph{gelation}. We characterise the gelation time $T_c$ by linking it to the criticality of the multitype branching process. Finally, in Corollary \ref{cor:localization}, we show that large clusters feature \emph{localization} -- their composition in terms of fractions of particles of different type concentrates on a specific vector, which we characterise using a variational problem. 
Section \ref{sec:mult_comp_smol} discusses the results and assumptions in details, while the proofs follow in Section \ref{sec:proofs}.


\section{Definitions and Results}\label{sec:mult_comp_smol}
Let $m \in \N$ and consider an $m \times m$ symmetric, irreducible matrix $A$. Let 
\begin{equation}
    \K(\mbf{k}, \mbf{l}) := \mbf{k}^T A \mbf{l}, \qquad \text{ for any } \mbf{k}, \mbf{l} \in \N_0^m
\end{equation}
be the coagulation kernel associated to $A$. We also require basic assumptions on the initial condition:
\begin{itemize}
    \item $w_{\mbf{n}}(0) \geq 0$ for all $\mbf{n} \in \N_0^m$, 
    \item $\sum_{\mbf{n} \in \N_0^m} |\mathbf{n}| w_{\mbf{n}}(0) = 1$, 
    \item $\{w_{\mbf{n}}(0)\}_{\mbf{n}\in \N_0^m}$ has non-empty, finite support.
\end{itemize}
Additionally, we require the initial conditions to be \textit{monodisperse}. This means that the initial distribution only consists of a mixture of `pure' components. That is $w_{\mbf{n}}(0)$ is non-zero only on the set of multi-indices $\mbf{n} = \mbf{e}_i$, where  $\mbf{e}_i$ is the $i$-th standard basis vector.  In this case, we define  $p_{i}:=w_{\mbf{e_i}}(0)$.


Having set the assumptions, we refer to equation \eqref{eq:colored_smoluchowski}  supplied with kernel $\K(\mbf{k}, \mbf{l})$ and initial condition $w_{\mbf{n}}(0)$ as the multicomponent Smoluchowski equation.
Next, we refer to the vector
\begin{equation}\label{eq:mass_vector}
    (\mbf{m}(t))_i := \sum\limits_{\mbf{n}\in \N_0^m} n_i w_{\mbf{n}}(t), \; i \in [m],
\end{equation}
as the \textit{total mass vector} of the multicomponent system at time $t$ and say that the equation \eqref{eq:colored_smoluchowski} exhibits gelation at time $T_c \in (0, \infty)$ such that $\mbf{m}(t) = \mbf{m}(0)$ for all $t \in [0, T_c)$ and $\mbf{m}_i(t) < \mbf{m}_i(0)$ for some $i$ and $t > T_c$.

Using the fact that the mass vector is conserved before gelation, allows us to simplify the coagulation equation \eqref{eq:colored_smoluchowski}  to obtain
\begin{equation}\label{eq:colored_smoluchowski_3}
\frac{\mathrm{d}w_{\mathbf{n}}}{\mathrm{d}t} = \frac{1}{2} \sum\limits_{\mathbf{k}+\mathbf{l} = \mathbf{n}} \K(\mathbf{k}, \mathbf{l}) w_{\mathbf{k}} w_{\mathbf{l}} - \mbf{n}^T A \mbf{m}(0) w_{\mbf{n}},  \text{ for all }\mathbf{n} \in \mathbb{N}_0^m,
\end{equation}
which we call the \textit{reduced multicomponent Smoluchowski equation}. 
\begin{proposition}\label{prop:reduced_equivalent}
The solutions of equations \eqref{eq:colored_smoluchowski} and \eqref{eq:colored_smoluchowski_3}
coincide for $t < T_c$ and both equations undergo gelation at $T_c$.
\end{proposition}
The proposition is proved in Section \ref{sec:proof_prop_reduced_equivalent}. Our task is now to analyse of the reduced equation \eqref{eq:colored_smoluchowski_3}, which we do by mapping it to a multidimensional PDE. 
Let $\{w_{\mbf{n}}(t)\}_{\mbf{n} \in \N_0^m}$ be the solution of the multicomponent Smoluchowski's equation \eqref{eq:colored_smoluchowski_3} before gelation. Consider the multivariate generating function (GF) 
\begin{equation}\label{eq:gf_transform}
    U(t, \mathbf{x}) := \sum\limits_{\mathbf{n} \in \mathbb{N}_0^m} w_{\mathbf{n}}(t) e^{-\mathbf{x}\cdot \mathbf{n}},
\end{equation}
where $\mathbf{x} \in [0, \infty)^m$ and $t \geq 0$. Let $\mathbf{u} = -\nabla U$, where $\nabla$ is the gradient with respect to the coordinates $x_1, \ldots, x_m$. We write $\nabla \mbf{u}$ for the Jacobian of $\mbf{u}$ with respect to the spatial variables $\mbf{x}$.
Finally, observe that
$
    \mbf{m}(t) = -\nabla U(t, \mbf{0}).
$
We have the following theorem:
\begin{theorem}\label{thm:gf_transform}
    The function $\mathbf{u}(t,\mbf{x})$, as defined above using the GF transform of the solution of the multicomponent coagulation equation \eqref{eq:colored_smoluchowski_3}, solves the following PDE:
    \begin{equation}\label{eq:colored_smoluchowski_pde}
        \frac{\partial}{\partial t} \mathbf{u}(t, \mathbf{x}) = -(\nabla \mathbf{u}(t, \mathbf{x})) A(\mathbf{u}(t, \mathbf{x})-\mathbf{m}(t)), \qquad \mathbf{x} \in [0, \infty)^m
    \end{equation}
    with initial condition $\mathbf{u}(0, \mathbf{x}) = -\nabla U (0, \mathbf{x})$ for $t < T_c$, where $T_c$ is the gelation time of \eqref{eq:colored_smoluchowski_3}. Moreover, any power series solution of \eqref{eq:colored_smoluchowski_pde} of the form
    \begin{equation}\label{eq:power_series_sol}
        \mathbf{u}(t, \mbf{x}) = \sum\limits_{\mbf{n} \in \N_0^m} \mbf{n} w_{\mbf{n}} e^{-\mbf{n} \cdot \mbf{x}}
    \end{equation}
    where $\{w_{\mbf{n}}(t)\}_{\mbf{n} \in \N_0^m}$ is a collection of coefficients,
    is a solution of \eqref{eq:colored_smoluchowski_3} with initial condition $\{w_{\mbf{n}}(0)\}_{\mbf{n} \in \N_0^m}$.
\end{theorem}
This theorem is proven in Section \ref{sec:proof_thm_gf_transform}. 
The theorem states that  if we manage to find a solution of the PDE \eqref{eq:colored_smoluchowski_pde} that has a power series expansion of the form \eqref{eq:power_series_sol}, we can transform it into a solution of the multicomponent Smoluchowski equation \eqref{eq:colored_smoluchowski_3}.
We will now show that such a power series solution is given by the \emph{overall total progeny} distribution of a certain \emph{multi-type branching process} determined by the initial conditions and the kernel matrix $A$.


First, we give the working definition of the branching process, as appears, for example, in \cite{harris_1963}. Let $X_{n,i}$ be independent and identically distributed copies of a non-negative discrete random variable $X$, indexed by $n,i\in \mathbb N_0$. We refer to the probability mass function $\mathbb P[X=k]$ as the offspring distribution. Then the recurrence equation 
$$Z_{n+1} = \sum_{i=1}^{Z_n} X_{n,i}$$
with $Z_0 = 1$ is called a \emph{simple branching process}, and the random variable
$T:=\sum_{n=0}^{\infty} Z_n$ -- its \emph{total progeny}.
A more general \emph{multi-type branching process} is defined in a similar fashion, by replacing the random variable $X$ with a collection of random vectors indexed with $k\in [m]:={1,2,\dots,m}$.
For fixed type $k$, let $\mathbf{X}_{k,(n,i)}$ be a collection of independent and identically distributed copies of a non-negative discrete random vector $\mathbf{X}_{k}$. The \emph{multi-type branching process started from type $i$} is the recurrence equation 
\begin{equation*}
    \mathbf{Z}_{n+1} = \sum\limits_{k=1}^m\sum\limits_{i=1}^{Z_{n, k}}\mathbf{X}_{k,(n, i)},
\end{equation*}
initialized with the standard unit vector, $\mathbf{Z}_0 = \mathbf{e}_i$, where $\mathbf{Z}_n = (Z_{n, 1}, Z_{n, 2}, \ldots,Z_{n, m})$. Then, the random variable $\mathbf{T}^{(i)} := \sum_{n=0}^\infty \mathbf{Z}_n$ is called the \emph{total progeny of the multi-type branching process started from type $i$}. By choosing the starting type $k$ at random with probability $p_k$, we define the \emph{overall total progeny} as $\mathbf{T} :=\sum_{k = 1}^m p_k \mathbf{T}^{(k)}$. 

In general, $\mathbf{T}$ is not a proper random variable, in the sense that it may take value $\infty$ with a positive probability. We then define  the extinction probability $\xi = \P(|\mathbf{T}|<\infty)$. When $\xi=1$ we say that  the branching process goes extinct almost surely (a.s.). Consider a matrix with elements $M_{ij} = \E[X_{i, j}]$.
 It can be shown \cite{harris_1963, athreya_1972} that the branching process goes extinct a.s. if and only if $\|M\|_2 \leq 1$. When $\|M\|_2 = 1$, the branching process is said to be \textit{critical}.



Now we are ready to define the specific multi-type branching process that will be used in the rest of the paper. Let the root have type $k \in [m]$ with probability $p_{k}$. We define the offspring vector $\mbf{X}_k$ by giving its probability generating function 
\begin{equation}
    G_{\mbf{X}_k}(\mbf{s}) = \prod\limits_{l=1}^m \exp(t A_{kl} p_l (s_l-1)),
\end{equation} 
incorporating the kernel matrix $A$ and the initial conditions.  Note that by varying the time parameter $t$, we can move between extinction and non-extinction of the branching process.
We define $T_c$ to be the time for which the branching process is critical. The overall total progeny of this branching process is denoted by $\mbf{T}$. The  total progeny conditioned on starting from type $k \in [m]$ is denoted by $\mbf{T}^{(k)}$. 
\begin{theorem}\label{thm:multitype_branching_process}
    Consider the multi-type branching process as defined above. The PDE \eqref{eq:colored_smoluchowski_pde} is solved by $\mbf{u}(t, \mbf{x})$ where $u_i(t, \mbf{x})$ are defined as
    \begin{equation}
        u_i(t, \mbf{x}) := p_i G_{\mbf{T}^{(i)}}(e^{-x_1}, \ldots, e^{-x_m})
    \end{equation}
    for each $i \in [m]$. This solution is smooth up until the critical time $T_c= \|A P\|_2^{-1}$, where $P = \mathrm{diag}(p_1, \ldots, p_m)$. 
\end{theorem}
The proof is given in Section \ref{sec:proof_thm_multitype_branching_process}.
\begin{corollary}\label{cor:gelation_time}
    The gelation time of the multicomponent Smoluchowski equation \eqref{eq:colored_smoluchowski_3} is given by the critical time $T_c$ of the multi-type branching process.
\end{corollary}
\begin{corollary}\label{cor:sol_multicomponent_smolochukowski}
    For any $\mbf{n} \in \N^m_0$, fix an arbitrary $i \in [m]$ such that $n_i>0$. The solution of the multicomponent Smoluchowski equation 
    \eqref{eq:colored_smoluchowski_3} at $\mbf{n}$ is given by
    \begin{equation}
        w_{\mbf{n}}(t) = \frac{p_i}{n_i} \P(\mbf{T}^{(i)} = \mbf{n})
    \end{equation}
    for $t \in (0, T_c).$ 
\end{corollary}
The corollary above can be combined with an arborescent form of Lagrange's inversion formula (see \cite{bender_1998}) to obtain a combinatorial formula for $w_{\mbf{n}}(t)$ that is similar to the results derived in \cite{kovchegov_2023, andreis_2023}. 

Finally, we discuss the concentration-like behaviour in the cluster size distribution $w_{\mbf{n}}(t)$ observed when $N:=|\mathbf{n}|$ is large.  In the multicomponent coagulation equation, each cluster is characterised by a vector $\bf n$, and we can interpret $N$ as the overall cluster size, irrespective of its composition.
Since large clusters can only be generated by a large number of merging events, the mass in the distribution should concentrate on the `mean' vector of all possible combinations if we condition $N$ to be large. This phenomenon,  called \textit{localization},  was shown to take place in other multicomponent coagulation equations that do not lead to gelation \cite{ferreira_2022,ferreira2021}, where the localization was characterised for asymptotically large time to be defined by the initial distribution.  We show that for the multiplicative kernel, localization generally depends on time, as well as the initial distribution and the kernel matrix. 

Let $P := \mathrm{diag}(p_1, \ldots, p_m)$, $\Delta_m:= \{\boldsymbol{\rho} \in \R^m : |\boldsymbol{\rho}| = 1\}$ and consider a function $\Gamma: \Delta_m \to \R$, 
\begin{equation}
    \Gamma(\rho_1, \ldots, \rho_m) = \sum\limits_{l=1}^m \left(\rho_l \log\left(\frac{\rho_l}{t(AP\boldsymbol{\rho})_l}\right) + t (AP\boldsymbol{\rho})_l\right) - 1.
\end{equation}
\begin{corollary}\label{cor:localization}
    Consider the multicomponent Smoluchowski equation \eqref{eq:colored_smoluchowski_3} with monodisperse initial conditions $p_i \neq 0$ for $i \in [m]$. For any $t \in (0, T_c)$ and stochastic vector $\boldsymbol{\rho}\in \Q^m$ with strictly positive elements,
    \begin{equation}
        \lim\limits_{N\to\infty}\frac{1}{N} \log w_{\bf n}(t)\Big|_{{\bf n}=N \boldsymbol{\rho}} = -\Gamma(\boldsymbol{\rho}).
    \end{equation}
    Moreover,  $\Gamma(\boldsymbol{\rho})$ is convex on $\Delta_m$.
\end{corollary}
Since $\Gamma(\boldsymbol{\rho})$ is convex, it attains a finite minimum ${\boldsymbol\rho}^*(t)=\min_{{\boldsymbol\rho} \in \Delta_m}\Gamma(\boldsymbol\rho)$, which generally depends on time $t$, monodisperse initial conditions $p_k$ and the kernel matrix $A$. The vector $\boldsymbol\rho^*$ gives the limiting configuration of clusters when their overall size $N$ tends to infinity. If $AP$ is a stochastic matrix, it can be shown that the localization direction $\boldsymbol\rho^*$ is a normalized eigenvector that does not depend on time, though this is not true in for general $AP$. The proof of the above corollaries is given in Section \ref{sec:proof_of_corollaries}.

\subsection{Relaxation of the assumptions}
 The requirements that $A$  has to be  be symmetric and irreducible are made without loss of generality. When $A$ is not irreducible, the system of ODEs can be broken into several independent subsystems that can be studied separately with out approach. This also means that there may be several critical points, one for each subsystems.
 For example, the kernel $A = I$ with $I$ being the identity matrix leads  to $m$ independent equations, which can be solved readily.
 Furthermore, even when $A$ is not symmetric, its action will summarised to $\frac{1}{2}(A+A^\top)$, as the (multidimensional) convolution is a commutative operation.
 
The assumption of monodisperse initial conditions is restrictive, although, we do see a possible avenue for relaxing it. Our branching process representation theory is applicable for initial conditions with full supports. For instance, in \cite{hoogendijk2023} we show that one can consider appropriately decaying initial conditions for the characteristic PDE of the monocomponent multiplicative Smoluchowski equation. Hence, in principle, it should be possible to repeat the proofs from this paper in the non-monodisprese setting, albeit with a cumbersome system for book keeping of the indices.

\section{Proofs}\label{sec:proofs}
\subsection{Proof of Proposition \ref{prop:reduced_equivalent}}\label{sec:proof_prop_reduced_equivalent}
\begin{proof}
By definition, $\mbf{m}(t) = \mbf{m}(0)$ for all $t \in (0, T_c)$ with $\mbf{m}$ being the mass vector of \eqref{eq:colored_smoluchowski}. Therefore, we can rewrite the second term in the RHS of \eqref{eq:colored_smoluchowski} as
\begin{equation*}
    \begin{split}
        \sum\limits_{\mbf{k} \in \mathbb{N}_0^m} \K(\mathbf{n}, \mathbf{k}) w_{\mathbf{n}} w_{\mathbf{k}} &= \sum\limits_{\mbf{k} \in \N_0^m} \mbf{n}^T A \mbf{k} w_{\mbf{n}} w_{\mbf{k}}
        = \mbf{n}^T A \mbf{m}(t) w_{\mbf{n}}.
    \end{split}
\end{equation*}
This shows that the reduced equation is equivalent to the original equation before gelation. 
Hence, if the reduced equation gels at $T_c$, so does the original equation. To show the reverse implication, suppose that the original equation gels at $T_c$ and that reduced equation gels at $T_c' > T_c$. Let $\mbf{m}'(t)$ be the mass vector of the reduced equation. By definition, $\mbf{m}'(t) = \mbf{m}'(0)$ for all $t \in (0, T_c')$ and $\mbf{m}'(0) = \mbf{m}(0)$. However, this means that we could turn the reduced equation back into the original equation, which contradicts the assumption that $T_c' > T_c$. 
\end{proof}

\subsection{Proof of Theorem \ref{thm:gf_transform}}\label{sec:proof_thm_gf_transform}
\begin{proof}
    The idea of the proof is to first derive a PDE for $U(t, \mbf{x})$ using the multicomponent Smoluchowski equation and then use the PDE for $U(t, \mbf{x})$ to derive a PDE for $\mbf{u}(t, \mbf{x})$.

    Consider $U(t, \mbf{x})$ as defined in \eqref{eq:gf_transform}. First, observe that 
    \begin{equation}\label{eq:gradient_U}
        \nabla U(t, \mbf{x}) = -\sum\limits_{\mbf{n}\in \N_0^m} \mbf{n} w_{\mbf{n}}(t) e^{-\mbf{x} \cdot \mbf{n}},
    \end{equation}
    where the gradient is taken with respect to $\mbf{x}$.
    Differentiating with respect to $t$ and substituting the right hand side of the multicomponent Smoluchowski equation with kernel $A$, we obtain
\begin{equation}\label{eq:deriv_t_U}
        \begin{split}
            \frac{\partial}{\partial t}U(t, \mbf{x}) &= \sum\limits_{\mbf{n} \in \N_0^m} \frac{\dint w_{\mbf{n}}}{\dint t} e^{-\mbf{x}\cdot \mbf{n}}\\
            &= \sum\limits_{\mbf{n}\in \N_0^m} \left(\frac{1}{2}\sum\limits_{\mbf{k}+\mbf{l} = \mbf{n}} \mbf{k}^T A \mbf{l} w_{\mbf{k}} w_{\mbf{l}} - \mbf{n}^T A \mbf{m}(0) w_{\mbf{n}}\right) e^{-\mbf{x} \cdot \mbf{n}}\\
            &= \underbrace{\frac{1}{2}\sum\limits_{\mbf{n}\in \N_0^m} \sum\limits_{\mbf{k}+\mbf{l} = \mbf{n}} \mbf{k}^T A \mbf{l} w_{\mbf{k}} w_{\mbf{l}} e^{-\mbf{x}\cdot \mbf{n}}}_{\boxed{W_1}} - \underbrace{\sum\limits_{\mbf{n}\in \N_0^m} \mbf{n}^T A \mbf{m}(0) w_{\mbf{n}} e^{-\mbf{x} \cdot \mbf{n}}}_{\boxed{W_2}}.
        \end{split}
    \end{equation}

    We will first look at term $\boxed{W_1}$. We use that $\mbf{k}+\mbf{l} = \mbf{n}$, the Cauchy product and the observation \eqref{eq:gradient_U} to obtain that
    \begin{equation}\label{eq:W1_term}
        \begin{split}
            \boxed{W_1} &= \frac{1}{2}\sum\limits_{\mbf{n} \in \N_0^m}\sum\limits_{\mbf{k}+\mbf{l} = \mbf{n}} (\mbf{k} w_{\mbf{k}} e^{-\mbf{x}\cdot \mbf{k}})^T A (\mbf{l}w_{\mbf{l}} e^{-\mbf{x}\cdot \mbf{l}})
            = \frac{1}{2}(\nabla U)^T A (\nabla U).
        \end{split}
    \end{equation}
    Next, we manipulate term $\boxed{W_2}$. Again, using observation \eqref{eq:gradient_U}, we obtain
    \begin{equation}\label{eq:W2_term}
        \begin{split}
            \boxed{W_2} &= \sum\limits_{\mbf{n} \in \N_0^m} (\mbf{n} w_{\mbf{n}} e^{-\mbf{x} \cdot \mbf{n}})^T A \mbf{m}(0)
            = - (\nabla U)^T A \mbf{m}(0).
        \end{split}
    \end{equation}
    Combining both \eqref{eq:W1_term} and \eqref{eq:W2_term} in \eqref{eq:deriv_t_U} results in the following PDE for $U(t, \mbf{x})$:
    \begin{equation}
        \frac{\partial}{\partial t} U(t, \mbf{x}) = \frac{1}{2}(\nabla U)^T A (\nabla U) + (\nabla U)^T A \mbf{m}(0).
    \end{equation}
    Next, use $\mbf{u} = - \nabla U$ to obtain
    \begin{equation}
        \begin{split}
            \frac{\partial \mbf{u}}{\partial t} &= - \nabla\left(\frac{\partial U}{\partial t}\right)\\
            &= - \nabla \left(\frac{1}{2}(\mbf{u}^T A \mbf{u}) - \nabla (\mbf{u}^T A \mbf{m}(0))\right)\\
            &= - ((\nabla \mbf{u}(t, \mbf{x})) A \mbf{u} - (\nabla \mbf{u}) A \mbf{m}(0))\\
            &= - (\nabla \mbf{u}) A (\mbf{u}-\mbf{m}(0)).
        \end{split}
    \end{equation}
    This finishes the first part of the theorem.

    Next, we will show that any power series solution of \eqref{eq:colored_smoluchowski_pde} of the form
    \begin{equation}\label{eq:cand_sol}
        \mbf{u}(t, \mbf{x}) = \sum\limits_{\mbf{n} \in \N_0^m} \mbf{n} w_{\mbf{n}} e^{-\mbf{n} \cdot \mbf{x}}
    \end{equation}
    solves the multicomponent Smoluchowski equation. By uniform convergence,
    we observe that
    \begin{equation}\label{eq:cand_sol_deriv}
        \nabla \mbf{u} = -\sum\limits_{\mbf{n} \in \N_0^m} \mbf{n} \mbf{n}^T w_{\mbf{n}} e^{-\mbf{n} \cdot \mbf{x}}.
    \end{equation}
    Therefore, by substituting \eqref{eq:cand_sol} and \eqref{eq:cand_sol_deriv} into the PDE \eqref{eq:colored_smoluchowski_pde} we obtain:
    \begin{equation}
        \begin{split}
            &(\nabla \mbf{u}) A \mbf{u}
            = -\sum\limits_{\mbf{k}\in\N_0^m} \sum\limits_{\mbf{l} \in \N_0^m} \mbf{k} \mbf{k}^T A \mbf{l} w_{\mbf{k}} w_{\mbf{l}} e^{-\mbf{k} \cdot \mbf{x}}\\
            &= -\frac{1}{2}\left(\sum\limits_{\mbf{k}\in\N_0^m} \sum\limits_{\mbf{l} \in \N_0^m} (\mbf{k} \mbf{k}^T) A \mbf{l} w_{\mbf{k}} w_{\mbf{l}} e^{-\mbf{k} \cdot \mbf{x}} e^{-\mbf{l} \cdot \mbf{x}} + \sum\limits_{\mbf{k}\in\N_0^m} \sum\limits_{\mbf{l} \in \N_0^m} (\mbf{l} \mbf{l}^T) A \mbf{k} w_{\mbf{k}} w_{\mbf{l}} e^{-\mbf{k} \cdot \mbf{x}} e^{-\mbf{l} \cdot \mbf{x}}\right),
        \end{split}
    \end{equation}
    where we have split the sum into two equal pieces and relabelled the vectors. Next, we use the symmetry of $A$ to obtain that
    \begin{equation}
            (\nabla \mbf{u}) A \mbf{u} = -\frac{1}{2}\left(\sum\limits_{\mbf{k}\in\N_0^m} \sum\limits_{\mbf{l} \in \N_0^m} (\mbf{k}+\mbf{l}) (\mbf{k}^T A \mbf{l}) w_{\mbf{k}} w_{\mbf{l}} e^{-\mbf{k} \cdot \mbf{x}} e^{-\mbf{l} \cdot \mbf{x}} \right).
    \end{equation}
    Partitioning the sum, we may write
    \begin{equation}
        (\nabla \mbf{u}) A \mbf{u} = -\frac{1}{2}\left(\sum\limits_{\mbf{n} \in\N_0^m} \mbf{n} \sum\limits_{\mbf{k} + \mbf{l} = \mbf{n}} \mbf{k}^T A \mbf{l} w_{\mbf{k}} w_{\mbf{l}} e^{-\mbf{n} \cdot \mbf{x}}\right).
    \end{equation}
    For the second term in \eqref{eq:colored_smoluchowski_pde}, substituting \eqref{eq:cand_sol_deriv} yields
    \begin{equation}
            (\nabla \mbf{u}) A \mbf{m}(0) = -\sum\limits_{\mbf{n} \in \N_0^m} \mbf{n} (\mbf{n}^T A \mbf{m}(0)) e^{- \mbf{n} \cdot \mbf{x}}.
    \end{equation}
    Combining all terms, we obtain the equality
    \begin{equation}
        \begin{split}
            \frac{\partial \mbf{u}}{\partial t} &= \sum\limits_{\mbf{n} \in \N_0^m} \mbf{n} \frac{\dint w_{\mbf{n}}}{\dint t} e^{-\mbf{n} \cdot \mbf{x}}\\
            &= \sum\limits_{\mbf{n} \in \N_0^m} \mbf{n} \left(\frac{1}{2}\sum\limits_{\mbf{k} + \mbf{l} = \mbf{n}} \mbf{k}^T A \mbf{l} w_{\mbf{k}} w_{\mbf{l}} - \mbf{n}^T A \mbf{m}(0)\right) e^{-\mbf{n} \cdot \mbf{x}}\\
            &= -(\nabla \mbf{u}) A (\mbf{u}-\mbf{m}(0)).
        \end{split}
    \end{equation}
    Therefore, 
    \begin{equation}
        0 = \sum\limits_{\mbf{n} \in \N_0^m} \mbf{n} e^{-\mbf{n} \cdot \mbf{x}} \left(\frac{\dint w_{\mbf{n}}}{\dint t} - \frac{1}{2} \sum\limits_{\mbf{k} + \mbf{l} = \mbf{n}} \mbf{k}^T A \mbf{l} + \mbf{n}^T A \mbf{m}(0) \right).
    \end{equation}
    It follows immediately that $\{w_{\mbf{n}}(t)\}_{\mbf{n} \in \N_0^m}$ satisfies the multicomponent Smoluchowski equation.
\end{proof}

\subsection{Proof of Theorem \ref{thm:multitype_branching_process}}\label{sec:proof_thm_multitype_branching_process}
Before proving the theorem, we formulate the following consequence of the implicit function theorem that is appropriate for our setting
\begin{lemma}
    Let $U \subseteq \R^{2m+1}$ be an open set and $\mbf{F}: U \to \R^m$ be a continuously differentiable map. Let $V \subseteq \R^{m+1}$ be an open set and consider the continuously differentiable map $\mbf{g}: V \to \R^m$ such that $\mbf{F}((t, \mbf{x}), \mbf{g}(t, \mbf{x})) = 0$. Then,
    \begin{equation}
        \frac{\partial \mathbf{g}}{\partial t}(t, \mathbf{x}) = - \left[\frac{\partial \mathbf{F}}{\partial \mathbf{g}}(t, \mathbf{x}, \mathbf{g}(t, \mathbf{x}))\right]_{m \times m}^{-1} \cdot \frac{\partial \mathbf{F}}{\partial t}(t, \mathbf{x}, \mathbf{g}(t, \mathbf{x}))
    \end{equation}
    and
    \begin{equation}
        J_{\mathbf{g}}(t, \mathbf{x}) = -\left[\frac{\partial \mbf{F}}{\partial \mbf{g}}(t, \mbf{x}, \mbf{g}(t, \mbf{x}))\right]_{m\times m}^{-1} \cdot \frac{\partial \mbf{F}}{\partial \mbf{x}}(t, \mbf{x}, \mbf{g}(t, \mbf{x})).
    \end{equation}
\end{lemma}

\begin{proof}[Proof of Theorem \ref{thm:multitype_branching_process}]
    Define $\mbf{G}(t, \mbf{x}) := (G_{\mbf{T}^{(1)}}, \ldots, G_{\mbf{T}^{(m)}})$, where $G_{\mbf{T}^{(i)}}$ is the probability generating function of $\mbf{T}^{(i)}$. Moreover, let $P := \mathrm{diag}(p_1, \ldots, p_m)$. Observe that
    \begin{equation}
        \mbf{u}(t, \mbf{x}) = P \mbf{G}(t, \mbf{x}).
    \end{equation}
    Therefore, it suffices to show that
    \begin{equation}\label{eq:red_pde}
        \frac{\partial \mbf{G}}{\partial t}(t, \mbf{x}) = -(\nabla \mbf{G}(t, \mbf{x})) A P (\mbf{G}(t, \mbf{x}) - \mathbf{1}),
    \end{equation}
    with $\mbf{1}$ being a vector of ones of length $m$, is solved by $\mbf{G}(t, \mbf{x})$ as defined above.

    From the theory of branching processes \cite{good_1955}, the following implicit system holds for vector $\mbf{G}(t, \mbf{x}) = (G_{\mbf{T}^{(1)}}, \ldots, G_{\mbf{T}^{(m)}})$:
    \begin{equation}\label{eq:implicit_system}
        \begin{split}
            G_{\mbf{T}^{(1)}} &= e^{-x_1} G_{\mbf{X}_1}(G_{\mbf{T}^{(1)}}, \ldots, G_{\mbf{T}^{(m)}})\\
            &\stackrel{\vdots}{\hphantom{=}}\\
            G_{\mbf{T}^{(m)}} &= e^{-x_m} G_{\mbf{X}_m}(G_{\mbf{T}^{(1)}}, \ldots, G_{\mbf{T}^{(m)}}).
        \end{split}
    \end{equation}
    We will make use of the implicit differentiation lemma above to show that the branching process solves the PDE \eqref{eq:red_pde}.

    Define $\mbf{F}: \R^{2m+1} \to \R^m$ to be the map given by
    \begin{equation}
        F_k(t, \mbf{x}, \mbf{y}) = y_k - e^{-x_k} G_{\mbf{X}_k}(y_1, \ldots, y_m), \qquad k \in [m].
    \end{equation}
    $\mbf{F}$ is continuous differentiable. Let $\mathbf{G}(t, \mathbf{x}) = (G_{\mbf{T}^{(1)}}, \ldots, G_{\mbf{T}^{(m)}})$ be as in the above Lemma. By the implicit function theorem, $\mbf{G}(t, \mbf{x})$ is continuously differentiable as long as $\frac{\partial \mbf{F}}{\partial \mbf{g}}$ is invertible. Furthermore, $\mbf{F}((t, \mbf{x}), \mbf{G}(t, \mbf{x})) = 0$ for all $t \in (0, T_c)$ and $\mbf{x} \in [0, \infty)^m$. Therefore, we can apply the above Lemma to implicitly differentiate \eqref{eq:implicit_system}. To this end, we calculate the partial derivatives of $\mbf{F}$ with respect to $t$ and $\mbf{x}$. For any $i \in [m]$, we have
    \begin{equation}\label{eq:der_F_t}
        \frac{\partial F_i}{\partial t}((t, \mbf{x}), \mbf{G}(t, \mbf{x})) = - e^{-x_i} G_{\mbf{X}^{(i)}}\sum\limits_{l=1}^m A_{il} p_l \left(G_{\mbf{T}^{(l)}}-1\right).\\
    \end{equation}
    Similarly, for any $i, j \in [m]$,
    \begin{equation}\label{eq:der_F_x}
        \begin{split}
            \frac{\partial F_i}{\partial x_j}((t, \mbf{x}), \mbf{G}(t, \mbf{x})) &= -\delta_{i, j} e^{-x_i} G_{\mbf{X}_i},
        \end{split}
    \end{equation}
    which is a diagonal matrix. Observe that by combining \eqref{eq:der_F_t} and \eqref{eq:der_F_x}, we obtain
    \begin{equation}
        \frac{\partial \mbf{F}}{\partial t}((t, \mbf{x}), \mbf{G}(t, \mbf{x})) = - \frac{\partial \mbf{F}}{\partial \mbf{x}} A (\mbf{u}(t, \mbf{x}) - \mbf{m}(0)).
    \end{equation}
    Multiplying both sides by $-\left[\frac{\partial \mbf{F}}{\partial \mbf{G}}((t, \mbf{x}), \mbf{G}(t, \mbf{x}))\right]_{m \times m}^{-1}$, we obtain
    \begin{equation}
        \frac{\partial \mbf{G}}{\partial t}((t, \mbf{x}), \mbf{G}(t, \mbf{x})) = -(\nabla \mbf{G}(t, \mbf{x})) A P (\mbf{G}(t, \mbf{x}) - \mbf{m}(0)),
    \end{equation}
    which finishes the proof of the first part of the Theorem.

    In order to show that the solution is smooth up until the critical time $T_c = \|AP\|_2^{-1}$, we look at 
    \begin{equation*}
        \begin{split}
        \frac{\partial \mbf{F}}{\partial \mbf{g}}(t, \mbf{x}, \mbf{G}(t, \mbf{x})) &= I - t A \, \mathrm{diag}(u_1(t, \mbf{x}), \ldots, u_m(t, \mbf{x}))\\
        &= I - t A P\, \mathrm{diag}(G_1(t, \mbf{x}), \ldots, G_m(t, \mbf{x})).
        \end{split}
    \end{equation*}
    Since $AP \,\mathrm{diag}(u_1(t, \mbf{x}), \ldots, u_m(t, \mbf{x}))$ is a matrix with exclusively non-negative entries, the above quantity is singular for $(t, \mbf{x})$ as soon as $t = \|A P \mbf{G}(t, \mbf{x})\|_2^{-1}$ by the Perron-Frobenius theorem. Since $\|\mbf{G}(t, \mbf{x})\|_2 \leq 1$, it follows that $T_c = \|AP\|_2^{-1}$ is the critical time.
\end{proof}

\subsection{Proofs of Corollaries \ref{cor:gelation_time}, \ref{cor:sol_multicomponent_smolochukowski} and \ref{cor:localization}}\label{sec:proof_of_corollaries}

\begin{proof}[Proof of Corollary \ref{cor:gelation_time}]
    Theorem \ref{thm:multitype_branching_process} shows that the gelation time of the multicomponent Smoluchowski equation is given by $T_c = \|AP\|_2^{-1}$. To show that it is equal to the critical time of the branching process, recall that $$G_{\mbf{X}_k}(\mbf{s}) = \prod\limits_{l=1}^m \exp(t A_{kl} p_l (s_l-1)).$$ Therefore, the matrix of expected offspring, denoted by $T_{\kappa}$ is given by
    \begin{equation}
        (T_{\kappa})_{i,j \in [m]} = t A_{ij} p_j.
    \end{equation}
    Results on multi-type branching processes (see \cite{harris_1963, athreya_1972}) say that a branching process is 
    critical when $\|T_{\kappa}\|_2 = 1$, which is equivalent to $t\|AP\|_2 = 1$.  The latter equation is satisfied at $t=T_c$.
\end{proof}

\begin{proof}[Proof of Corollary \ref{cor:sol_multicomponent_smolochukowski}]
    Combining Theorem \ref{thm:multitype_branching_process} and the second part of Theorem \ref{thm:gf_transform}, we immediately obtain the result of the corollary.
\end{proof}
Before proving Corollary \ref{cor:localization}, we introduce the following lemma, which is a consequence of the Lagrange-Good inversion formula \cite{good_1960}. 
\begin{lemma}\label{lem:lagrange_inversion_cor}
    Given a multi-type branching process satisfying the implicit system of equations \eqref{eq:implicit_system}, for any $\mbf{n} \in \N^m$ and $i \in [m],$
    \begin{equation}
        \P(\mbf{T}^{(i)} = \mbf{n}) = [\mbf{r}^{\mbf{n}}]r_i \mathrm{det}(K(\mbf{r})) (G_{\mbf{X}_1}(\mbf{r}))^{n_1} \cdot \ldots \cdot (G_{\mbf{X}_m}(\mbf{r}))^{n_m},
    \end{equation}
    where 
    \begin{equation}
        K(\mbf{r}):=\left[\delta_{ij} - \frac{r_i}{G_{\mbf{X}_i}(\mbf{r})} \frac{\partial G_{\mbf{X}_i}}{\partial r_j}\right]_{1 \leq i, j \leq m}.
    \end{equation}
\end{lemma}
\begin{proof}[Proof of Corollary \ref{cor:localization}]
    Note that by Corollary \ref{cor:sol_multicomponent_smolochukowski}, we have that for any $\mbf{n} \in \N^m$ and $i \in [m]$,
    \begin{equation}
        w_{\mbf{n}} = \frac{p_i}{n_i}\P(\mbf{T}^{(i)} = \mbf{n}).
    \end{equation}
    We will analyze $\P(\mbf{T}^{(i)} = \mbf{n})$ further. Using Lemma \ref{lem:lagrange_inversion_cor}, we obtain that
    \begin{equation}\label{eq:tot_prog_prob_1}
        \P(\mbf{T}^{(i)} = \mbf{n}) = [\mbf{r}^{\mbf{n}}]r_i \mathrm{det}(K(\mbf{r})) G_{\mbf{X}_1}^{n_1} \cdot \ldots \cdot G_{\mbf{X}_m}^{n_m}.
    \end{equation}
    It can be readily verified that 
    \begin{equation}
        K(\mbf{r}) = [\delta_{ij} - r_i t A_{ij} p_j]_{1 \leq i, j \leq m}.
    \end{equation}
    Moreover, notice that $\det K(\mbf{r})$ can be written as
    \begin{equation}
        \det K(\mbf{r}) = \sum\limits_{I \in \mathcal{P}([m])} c_I r_I,
    \end{equation}
    where $r_I = \prod_{i\in I} r_i$ and $c_I$ being a coefficient dependent on $t$, $A$ and $\{p_i\}_{i \in [m]}$. Therefore, we may write \eqref{eq:tot_prog_prob_1} as 
    \begin{equation}\label{eq:tot_prog_prob_2}
        \P(\mbf{T}^{(i)} = \mbf{n}) = \sum\limits_{I \in \mathcal{P}([m])} c_I [\mbf{r}^{\mbf{n}}]r_I r_i G_{\mbf{X}_1}^{n_1} \cdot \ldots \cdot G_{\mbf{X}_m}^{n_m}.
    \end{equation}
    For the next step in our analysis, let $Z_l \sim \mathrm{Poi}\left(t \sum_{k=1}^m n_k A_{kl} p_l\right)$ for each $l \in [m]$. By defining
    \begin{equation}
        \sigma_l:= \sum\limits_{k=1}^m \rho_k A_{kl} p_l,
    \end{equation}
    we see that $Z_l \sim \mathrm{Poi}\left(t N \sigma_l\right)$. Observe that the probability generating function of $Z_l$ is given by
    \begin{equation}
        G_{Z_l}(r) = \exp\left(t N \sigma_l (r-1)\right)
    \end{equation}
    Then, we note that the product $G_{\mbf{X}_1}(\mbf{r})^{n_1} \cdot \ldots \cdot G_{\mbf{X}_m}(\mbf{r})^{n_m}$ can be written as
    \begin{equation}\label{eq:factorization_pgfs_multi}
        \begin{split}
            G_{\mbf{X}_1}(\mbf{r})^{n_1} \cdot \ldots \cdot G_{\mbf{X}_m}(\mbf{r})^{n_m} &= \prod\limits_{k=1}^m \prod\limits_{l=1}^m \exp\left(t n_k A_{kl} p_l (r_l-1)\right)\\
            &= \exp\left(t \sum\limits_{l=1}^m \sum\limits_{k=1}^m n_k A_{kl} p_l (r_l-1)\right)\\
            &= \prod\limits_{l=1}^m \exp\left(t N \sigma_l (r_l-1)\right)\\
            &= \prod\limits_{l=1}^m G_{Z_l}(r_l).
        \end{split}
    \end{equation}
    Combining \eqref{eq:tot_prog_prob_2} and \eqref{eq:factorization_pgfs_multi} yields
    \begin{equation}
        \begin{split}
            \P(\mbf{T}^{(i)} = \mbf{n}) &= \sum\limits_{I \in \mathcal{P}([m])} c_I [\mbf{r}^{\mbf{n}}] r_I r_i G_{Z_1}(r_1) \cdot \ldots \cdot G_{Z_m}(r_m)\\
            &= \sum\limits_{I \in \mathcal{P}([m])} c_I [\mbf{r}^{\mbf{n}}] \left(\prod\limits_{l=1}^m r_l^{\mathbbm{1}_I(l) + \delta_{li}}\right) G_{Z_1}(r_1) \cdot \ldots \cdot G_{Z_m}(r_m)\\
            &= \sum\limits_{I \in \mathcal{P}([m])} c_I \prod\limits_{l=1}^m [r_l^{n_l}] r_l^{\mathbbm{1}_I(l) + \delta_{li}} G_{Z_l}(r_l)\\
            &= \sum\limits_{I \in \mathcal{P}([m])} c_I \prod\limits_{l=1}^m \P(Z_l = n_l - (\mathbbm{1}_I(l) + \delta_{li}))
        \end{split}
    \end{equation}
    The last line implies that we need to analyze 
    \begin{equation}
        \P(Z_l = n_l), \qquad \P(Z_l = n_l-1), \qquad \P(Z_l = n_l - 2).
    \end{equation}
    We are going to relate the latter two terms to the first term, making explicit use of the fact that $Z_l \sim \mathrm{Poi}\left(t N \sigma_l\right)$ for each $l \in [m]$. We obtain
    \begin{equation}
        \P(Z_l = n_l-1) = e^{-t N \sigma_l} \frac{(t N \sigma_l)^{n_l-1}}{(n_l-1)!} = \frac{\rho_l}{t \sigma_l} \P(Z_l = n_l)
    \end{equation}
    and
    \begin{equation}
        \P(Z_l = n_l-2) = e^{-\lambda_l} \frac{\lambda_l^{n_l-2}}{(n_l-2)!} = \left(\frac{\rho_l^2}{t^2 \sigma_l^2} - \frac{\rho_l}{t^2 N \sigma_l^2}\right) \P(Z_l = n_l)
    \end{equation}
    Therefore, we obtain that
    \begin{equation}
        \log \P(\mbf{T}^{(i)} = \mbf{n}) = \sum\limits_{l=1}^m \log \P(Z_l = n_l) + \log \sum\limits_{I \in \mathcal{P}([m])} c_I \mathcal{O}\left(1-\frac{1}{N}\right).
    \end{equation}
    Observe that the second term will tend to zero as we divide by $N$ and sending $N \to \infty$. Using the explicit formula of the Poisson distribution together with Stirling's formula $$ n! = (1+o(1)) \sqrt{2\pi n} n^n e^{-n},$$ we obtain
    \begin{equation}
        \begin{split}
            \sum\limits_{l=1}^m \log \P(Z_l = N \rho_l) &= \sum\limits_{l=1}^m \log\left(\frac{(t N \sigma_l)^{N \rho_l}}{(N \rho_l)!} e^{-t N \sigma_l}\right)\\
            &= N \sum\limits_{l=1}^m  \left(\rho_l \log(t N \sigma_l) - \rho_l \log(N \rho_l) + \rho_l \right)+ \mathcal{O}\left(\log(N)\right)\\
            &= N \sum\limits_{l=1}^m
            \left(\rho_l \log\left(\frac{t \sigma_l}{\rho_l}\right) - t \sigma_l + \rho_l\right) + 
            \mathcal{O}\left(\log(N)\right)
        \end{split}
    \end{equation}
    Dividing both sides by $N$ and sending $N \to \infty$ finishes the first part of the proof. 
    
    It remains to prove that 
    \begin{equation}
        \Gamma(\mathbf{\rho}) = \sum\limits_{l=1}^m \rho_l \log\left(\frac{\rho_l}{t \sigma_l}\right) + t \sigma_l - 1
    \end{equation}
    is convex. The idea of the proof is to use objects that are very reminiscent of objects that also appear in large devations theory. Consider the function 
    \begin{equation}
        \Lambda_l(\boldsymbol{\rho}) = \sup\limits_{\lambda \in \R} \{\lambda \rho_l - \sigma_l(e^{\lambda}-1)\}.
    \end{equation}
    Now, we claim that the equality $\Gamma(\boldsymbol{\rho}) = \sum_{l=1}^m \Lambda_l(\boldsymbol{\rho})$ holds and  that $\Lambda_l(\boldsymbol{\rho})$ is convex for each $l \in [m]$. The first claim follows immediately by differentiating the argument within the supremum of $\Lambda_l$, equating  to zero and solving for $\lambda$. It can be readily checked that this yields a unique maximumum. The second claim follows by considering $\boldsymbol{\rho}^{(1)}, \boldsymbol{\rho}^{(2)} \in \Delta_m$ and $\nu \in [0, 1]$. By using subadditivity of the supremum, we obtain that
    \begin{equation}
        \Lambda_l(\nu \boldsymbol{\rho}^{(1)} + (1-\nu) \boldsymbol{\rho}^{(2)}) \leq \nu \Lambda_l(\boldsymbol{\rho}^{(1)}) + (1-\nu) \Lambda_l(\boldsymbol{\rho}^{(2)}).
    \end{equation}
    Therefore, $\Lambda_l$ is convex for each $l \in [m]$. Since $\Gamma$ is a sum of convex functions, it is convex as well.
\end{proof}

\section*{Acknowledgements}
This publication is part of the project ``Random graph representation of nonlinear evolution problems"  of the research programme Mathematics Cluster/NDNS+ which is financed by the Dutch Research Council (NWO). The authors thank Mike de Vries for his help with the proof of the Corollary \ref{cor:localization}.

\bibliographystyle{alpha}
\bibliography{references.bib}

\end{document}